\documentclass[a4paper,UKenglish]{article}

\pdfoutput=1

\newcommand{\cslonly}[1]{}

\usepackage[utf8]{inputenc}
\usepackage{amsmath}
\usepackage{amsthm}
\usepackage{amssymb}
\usepackage{tikz}
\usetikzlibrary{cd,arrows,decorations.pathmorphing}
\usepackage{todonotes}
\usepackage{multicol}

\usepackage[backend=biber,safeinputenc,giveninits=true,hyperref=true]{biblatex}
\usepackage[pdfusetitle]{hyperref}

\usepackage{mathtools}
\usepackage[T1]{fontenc}

\usepackage{mathpartir}

\newcommand{\set}[1]{\left\{ \, #1 \,\right\}}

\newcommand{\pTrunc}[1]{\left\| #1 \right\|}

\newcommand{\abs}[1]{\left| #1 \right|}

\newcommand{\parens}[1]{\left( #1 \right)}
\makeatletter
\def\lam#1{{\lambda}\@lamarg#1:\@endlamarg\@ifnextchar\bgroup{.\,\lam}{.\,}}
\def\@lamarg#1:#2\@endlamarg{\if\relax\detokenize{#2}\relax #1\else\@lamvar{\@lameatcolon#2},#1\@endlamvar\fi}
\def\@lamvar#1,#2\@endlamvar{(#2\,{:}\,#1)}
\def\@lameatcolon#1:{#1}

\def\lamu#1{{\lambda}\@lamuarg#1:\@endlamuarg\@ifnextchar\bgroup{.\,\lamu}{.\,}}
\def\@lamuarg#1:#2\@endlamuarg{#1}

\newcommand{\journalonly}[1]{}
\newcommand{\UU}{\mathcal{U}} 
\newcommand{\VV}{\mathcal{V}} 
\newcommand{\Prop}{\mathrm{HProp}} 

\newcommand{\bool}{\mathbf{2}} 
\newcommand{\true}{\mathsf{t\!t}} 
\newcommand{\false}{\mathsf{f\!f}} 
\newcommand{\rat}{\operatorname{{rat}}} 
\newcommand{\eq}{\operatorname{\mathsf{eq}}} 
\renewcommand{\lim}{\operatorname{{lim}}} 
\newcommand{\sm}[1]{{(\Sigma{#1})}} 
\newcommand{\dpt}[1]{{(\Pi{#1})}} 
\newcommand{\ex}[1]{{(\exists{#1})}} 
\newcommand{\fa}[1]{{(\forall{#1})}} 
\newcommand{\pow}{\mathcal{P}} 
\newcommand{\N}{\mathbb{N}} 
\newcommand{\Q}{\mathbb{Q}} 
\newcommand{\Z}{\mathbb{Z}} 
\newcommand{\R}{\mathbb{R}} 
\newcommand{\RC}{{\mathbb{R}_\mathbf{C}}} 
\newcommand{\Rc}{{\mathbb{R}_c}} 
\newcommand{\RH}{{\mathbb{R}_\mathbf{H}}} 
\newcommand{\RD}{{\mathbb{R}_\mathbf{D}}} 
\newcommand{\RE}{{\mathbb{R}_\mathbf{E}}} 
\newcommand{\isContr}{\operatorname{{isContr}}} 
\newcommand{\CA}{\mathcal{C}} 
\newcommand{\isCauchy}{\operatorname{{isCauchy}}} 

\newcommand{\refl}{\operatorname{\mathsf{refl}}} 
\newcommand{\ap}{{\operatorname{{ap}}}} 
\newcommand{\proj}[1]{\operatorname{{pr}_{#1}}} 
\newcommand{\hottbook}{\textcite{hottbook}}
\newcommand{\eqv}{\simeq} 
\newcommand{\Iff}{\Leftrightarrow} 
\newcommand{\Implies}{\Rightarrow} 
\newcommand{\CShom}{\operatorname{{CS-hom}}} 
\newcommand{\subCShom}{\operatorname{{sub-CS-hom}}} 

\theoremstyle{plain}
\newtheorem{theorem}{Theorem}
\newtheorem*{theorem*}{Theorem}
\newtheorem{lemma}[theorem]{Lemma}
\newtheorem{proposition}[theorem]{Proposition}
\newtheorem{corollary}[theorem]{Corollary}

\theoremstyle{definition}
\newtheorem{definition}[theorem]{Definition}

\newtheorem*{definition*}{Definition}

\newtheorem*{example*}{Example}
\newtheorem{claim}{Claim}
\theoremstyle{remark}
\newtheorem{remark}[theorem]{Remark}

\providecommand{\claimproofname}{Proof of claim}

\usepackage{hyphenat}

\title{The HoTT book reals coincide with the Euclidean reals}
\date{}
\author{Auke B. Booij}

\begin{document}

\maketitle

\begin{abstract}
  Escard\'o and Simpson defined a notion of interval object by a universal
  property in any category with binary products. The Homotopy Type Theory book
  defines a higher inductive-inductive notion of reals, and suggests that the
  interval in this type may satisfy this universal property. We show that this
  is indeed the case in the category of sets of any universe. We also show that
  the type of HoTT reals is the smallest Cauchy complete subset of the Dedekind
  reals containing the rationals.
\end{abstract}

\section{Introduction}
\label{cha:univalent:math}
Escard\'o and Simpson introduced the notion of \emph{interval object},
which can be defined in any category with finite products, as a
universal property for closed and bounded real line
segments~\cite{escardo:simpson:interval}.  Indeed, in the category of
classical sets, the real interval $[-1,1]$ is an interval object.  In
the category of topological spaces, the real interval $[-1,1]$ with
the Euclidean topology is an interval object.
\Textcite{vickers:interval} showed that in the category of
locales, the locale corresponding to the interval $[-1,1]$ is an
interval object.

In a topos, the interval $[-1,1]$ in a certain subobject $\RE$ of the
Dedekind reals is an interval object.  The object $\RE$, referred to
as the Euclidean reals, is defined as the least Cauchy complete subset
of the Dedekind reals containing the rationals.  This can be
constructed as the intersection of all Cauchy complete subsets of the
Dedekind reals that contain the rationals.

Assuming the propositional resizing axiom of
Definition~\ref{def:prop-resiz}, we can translate the construction of the
Euclidean reals $\RE$ as an intersection of subsets of $\RD$ into type
theory, and similarly translate the proof that the interval in $\RE$
is an interval object.  The fact that $\RE$ is the least Cauchy
complete subset of $\RD$ containing the rationals is then easily
verified.  The Euclidean reals sit between the Cauchy reals and the
Dedekind reals: we have the sequence of canonical inclusions
\[
  \RC\subseteq\RE\subseteq\RD
\]
where neither of the inclusions can be shown to be an equality.  This reminds us
of the HoTT book reals---whose definition we recall in
Section~\ref{sec:reals:hott}---which also sits between $\RC$ and $\RD$ in a
canonical way:
\[
  \RC\subseteq\RH\subseteq\RD.
\]

This raises the question whether the HoTT book reals and the Euclidean
reals coincide, so that the interval in the HoTT book reals would be
an interval object.
\textcite[Chapter~11,~Notes]{hottbook} indeed conjectures that $\RH$
and $\RE$ coincide.

\begin{quote}
  The fact that $\Rc$ is the least Cauchy complete
  archimedean ordered field, as was proved in Theorem~11.3.50,
  indicates that our Cauchy reals probably coincide with the
  Escard\'o-Simpson reals.  It would be interesting to check whether
  this is really the case.

  --- \textit{\textcite[Chapter~11,~Notes]{hottbook}.  Note that we
    introduce this object $\Rc$ as the ``HoTT book reals, $\RH$'' in
    Definition~\ref{def:hott:book:reals}.}
\end{quote}

When phrasing this question more precisely, we are reminded that we
may be working in a type theory in which we do not have the HoTT book
reals, or in a type theory which does not have propositional resizing,
so that we cannot construct the Euclidean reals.

We can relate $\RH$ to $\RE$ by showing that $\RH$ is the least Cauchy
complete subset of $\RD$ containing the rationals, as we do in
Section~\ref{sec:reals:euclidean}.  This result can be phrased without
propositional resizing, since $\RE$ can be characterized as the least
Cauchy complete subset of the Dedekind reals containing the
rationals. In particular, when we do have propositional resizing, we
can construct $\RE$ and, by this result, it coincides with $\RH$.

If we do not have $\RH$, we can still relate $\RE$ to $\RH$ by showing that
$\RE$ satisfies a universal property similar to the one of the HoTT book reals
given in Definition~\ref{def:hott:book:reals}.  This result in
Section~\ref{sec:assuming:prop} assumes propositional resizing.  In particular,
when we do have $\RH$, this result implies that $\RE$ and $\RH$ coincide.

We use propositional resizing to construct $\RE$, and by the above, we can prove
it has a certain universal property.  We may also wonder whether a least Cauchy
complete subset of the Dedekind reals containing the rationals, without knowing
the construction of $\RE$ as an intersection of subsets, has this universal
property.  This question can be phrased without propositional resizing, but we
do not answer it.

In summary, we confirm the conjecture of \hottbook\ in two ways, once assuming
the existence of $\RH$, and once assuming propositional resizing.  We leave an
open question in the absence of both $\RH$ and propositional resizing.

\section{Preliminaries}
\label{sec:preliminaries}

We use Martin-L\"of Type Theory with univalence, function extensionality,
quotient types, propositional truncation, higher inductive-inductive types, and
propositional resizing.

We write $\UU$ for a univalent universe.  Given a type $X$, we write
$\pTrunc{X}$ for its propositional truncation.  By ``there exists $a:A$ such
that $B(a)$'' we mean $(\exists(a:A).B(a)) \coloneqq \pTrunc{\sm{a:A}B(a)}$,
and by the disjunction $X \vee Y$ we mean $\pTrunc{X+Y}$.  An \emph{equivalence}
$e:X \eqv Y$ between two types $X$ and $Y$ is given by a map $f:X \to Y$ which
has both a left inverse and a right inverse, and we implicitly coerce equivalences
$X \eqv Y$ to their underlying map $X \to Y$.  We write $\Prop_\UU$ for the
type of propositions in universe $\UU$, or just $\Prop$ if we leave the universe
implicit.

\subsection{Subtypes and embeddings}
\label{sec:subtypes}

\begin{definition}\label{def:subtype}
  By a \emph{$\VV$-subtype} $B:\pow_\VV A$ of $A:\UU$ we mean a map
  $B:A\to\Prop_\VV$.  For $b:A$ we define $(b\in B) \coloneqq B(b)$.  We say
  \emph{subtype}, denoted $B:\pow A$, if we wish to leave the universe $\VV$
  implicit.  A \emph{subset} is a subtype of a type that is a
  set~\cite[Definition~3.1.1]{hottbook}.
\end{definition}
This is motivated by the fact that if $B:\pow_\VV A$ is a subtype
of $A$, then the projection map $\proj1:\sm{a:A}B(a)\to A$ is an
\emph{embedding}, and vice versa embeddings give rise to subtypes, as
we will make precise in Lemma~\ref{lem:embeddings-subtypes}.

\begin{definition}\label{def:embeddings}
  Given a function $f:C\to A$, we say $f$ is an
  \emph{embedding}, and write $f:C\hookrightarrow A$,
  if $\ap_{f,c,c'}:(c=_{C}c')\to(fc=_{A}fc')$ is an equivalence
  for all $c,c':C$.
\end{definition}

Definitions~\ref{def:subtype} and~\ref{def:embeddings} are equivalent
in the following sense.
\begin{lemma}\label{lem:embeddings-subtypes}
  A subtype $B:\pow_\VV A$ of $A:\UU$ gives rise to a type
  $C:\UU\sqcup\VV$ that embeds into $A$, where $\UU\sqcup\VV$ is the
  least universe above $\UU$ and $\VV$.  Conversely, a type
  $C:\VV$ with an embedding into $A:\UU$ gives rise to a subtype
  $B:\pow_{\UU\sqcup\VV}A$.  These constructions are inverse to each
  other.
\end{lemma}
\begin{proof}
  In one direction, the type $C\coloneqq\sm{b:A}b\in B$ embeds into
  $A$ by the projection map.  Conversely, given an embedding
  $f:C\hookrightarrow A$, the subtype is given by
  $B(a)\coloneqq\sm{c:C} (f c = a)$, which is well-defined by the fact
  that $f$ is an embedding.  For details, see e.g.\
  \textcite[Theorem~3.29]{RijkeSpitters:2015:Sets-In-HoTT}.
\end{proof}
\begin{remark}
  This result uses univalence to show an equality of types.
\end{remark}
We will often use this correspondence implicitly.

\begin{lemma}\label{cor:subtypes-poset}
  For $A:\UU$, the type $\pow_\VV A$ is a partial order with $\subseteq$.
  Explicitly, with $P,Q,R:\pow_\VV A$:
  \begin{enumerate}
  \item $P\subseteq P$,
  \item $(P\subseteq Q) \to (Q\subseteq P) \to P = Q$,
  \item $(P\subseteq Q) \to (Q\subseteq R) \to (P\subseteq R)$.
  \end{enumerate}
\end{lemma}
\begin{proof}
  Straightforward, where antisymmetry uses function extensionality and propositional extensionality.
\end{proof}

Section~\ref{sec:reals:euclidean} uses the following formulation of
the antisymmetry of Lemma~\ref{cor:subtypes-poset}.  Note that for
embeddings from $C$ and $D$ into $A$, the relation $C\subseteq D$
holds if we have a certain commutative triangle as below.

\begin{lemma}\label{lem:triangle:embeddings}
  Suppose given a triangle of maps as follows.

  
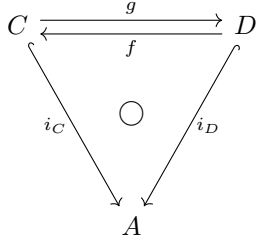
\begin{tikzcd}[column sep={1.5cm,between origins}]
    C
    \arrow[rr,yshift=0.6ex,"f"]
    \arrow[ddr,hookrightarrow,"i_C" left]
    &
    &
    D
    \arrow[ddl,hookrightarrow,"i_D" right]
    \\ & \node[yshift=1cm]{\bigcirc}; &
    \\
    &
    A
    &
  \end{tikzcd}

  If $i_C$ and $i_D$ are embeddings, and the commutativity condition
  $i_D\circ f = i_C$ is satisfied, then $f$ is an embedding.
\end{lemma}
\begin{proof}
  We can show, by induction on identity types, that
  $\ap_{i_C}=\ap_{i_D}\circ\ap_f$.  Then, by a two-out-of-three
  property for equivalences~\cite[Theorem~4.7.1]{hottbook}, we get
  that $\ap_f$ is an equivalence, as required.
\end{proof}

\subsection{Cauchy structures}
\label{sec:cauchy-structures}

Let $\N$ and $\Z$ be appropriate types of naturals and integers.  We can define
a type $\Q$ of rationals as in~\textcite[Section~11.1]{hottbook} with
their ordering $<$.  It will be convenient to additionally define the type of
positive rationals:
\[
  \Q_+\coloneqq\set{q:\Q \mid q>0 }=\sm{q:\Q}(q>0).
\]

Following \textcite{sojakova:hit}, we take an algebraic view on types
of real numbers and higher inductive-inductive types (HIITs).  We will
do this, as opposed to directly giving the type-theoretic inference
rules of the HIIT, in order to make a clear link with the Euclidean
reals in Section~\ref{sec:reals:euclidean}.  By analogy with
\textcite{MALQ:MALQ200710024} and the \hottbook, we define premetric
spaces.
\begin{definition}
  A \emph{premetric} on a type $R:\UU$ is a relation
  \[
    \cdot\sim_\cdot\cdot:R \times \Q_+\times R \to \Prop.
  \]
  We will often write $R$ for the \emph{premetric space}
  $(R,\sim)$, leaving the premetric $\sim$ implicit.
\end{definition}
Note the outright lack of natural conditions one might put on $\sim$: our
premetric spaces are a much wilder notion than Richman's, lacking even something
as basic as a triangle inequality.  In fact, having few conditions here is a
good thing, as any conditions introduced now would need to be respected later by
the induction principle in Definition~\ref{def:hott:book:reals}, thus making
that induction principle harder to use.

\begin{definition}
\label{def:cauchy-approximation}
  If $R$ is a premetric space, then $x:\Q_+\to R$ is a
  \emph{Cauchy approximation} if
  \begin{equation}
    \isCauchy(x)\coloneqq\fa{\delta,\varepsilon:\Q_+}
    x_\delta\sim_{\delta+\varepsilon}x_\varepsilon.\label{eq:metric:cauchy:builtin}
  \end{equation}
  We define the type $\CA_{R}$ of Cauchy approximations
  in $R$ as
  \[
    \CA_{R}\coloneqq\sm{x:\Q_+\to R} \isCauchy(x).
  \]
  Since being a Cauchy approximation is a property rather than
  structure, we implicitly coerce elements of $\CA_{R}$ to their underlying
  map $\Q_+\to R$.
\end{definition}

\begin{definition}\label{def:limits:premetric}
  If $x$ is a Cauchy approximation in a premetric space $R$,
  then we say that $u:R$ is a \emph{limit} of $x$ if
  \[
    \fa{\varepsilon,\theta:\Q_+}
    x_\varepsilon\sim_{\varepsilon+\theta}u.
  \]
  If there exists a limit for every Cauchy approximation, we say that $R$ is
  \emph{Cauchy complete}.
\end{definition}
For the types we will in any regard consider Cauchy complete, limits will be
unique, so that we can always compute them.

In our very weak notion of premetric spaces, we do not automatically have
uniqueness of limits, so that the \emph{existence} of limits does not imply that
we can \emph{compute} limits.  Instead, we will simply assume that we have a
$\lim$ map for this purpose.

\begin{definition}\label{def:cauchy:structure}
  A \emph{Cauchy structure} is a premetric space $(R,\sim)$ together
  with the following structure, collected in a $\Sigma$-type.
  \begin{align*}
    \rat:\;&\Q\to R \\
    \lim:\;&\CA_{R}\to R \\
    \eq:\;& \dpt{u,v:R} \parens{ \fa{\varepsilon:\Q_+}  u
              \sim_\varepsilon v} \to u =_R v
    \\
    d_{\rat,\rat}:\;& \dpt{q,r:\Q}\dpt{\varepsilon:\Q_+}
      (-\varepsilon<q-r<\varepsilon) \to \rat(q)\sim_\varepsilon
      \rat(r) \\
    d_{\rat,\lim}:\;& \dpt{q:\Q}\dpt{y:\CA_{R}}
      \dpt{\varepsilon,\delta:\Q_+}
      \rat(q)\sim_{\varepsilon}y_\delta\to
      \rat(q)\sim_{\varepsilon+\delta}\lim(y) \\
    d_{\lim,\rat}:\;& \dpt{x:\CA_{R}}\dpt{r:\Q}
      \dpt{\varepsilon,\delta:\Q_+}
      x_\delta\sim_{\varepsilon}\rat(r)\to
      \lim(x)\sim_{\varepsilon+\delta}\rat(r) \\
    d_{\lim,\lim}:\;& \dpt{x:\CA_{R}}\dpt{y:\CA_{R}}
      \dpt{\varepsilon,\delta,\eta:\Q_+}
      x_\delta \sim_{\varepsilon} y_\eta\to
      \lim(x) \sim_{\varepsilon+\delta+\eta} \lim(y)
  \end{align*}
  A morphism of Cauchy structures from $R$ to $S$ is a
  map $f:R\to S$ and a family of maps
  $g_{\varepsilon,u,v}:u\sim_\varepsilon v \to f(u)\sim_\varepsilon
  f(v)$ that preserve $\rat$, $\lim$ and $\eq$ in the obvious
  sense.  Explicitly:
\begin{align*}
  \CShom(R,S) \coloneqq
  & \sm{f:R\to S}\\
  & \sm{g:\dpt{u,v:R} \dpt{\varepsilon:\Q_+}u\sim_\varepsilon v
    \,\to\, f(u)\sim_\varepsilon f(v)} \\
  & \left(\dpt{q:\Q} f(\rat(q))=\rat(q)\right) \\
  \times
  & \left(\dpt{x:\CA_R} f(\lim(x))=\lim(f \circ x)\right)\\
  \times
  & \big(\dpt{u,v:R}
    \dpt{p:\fa{\varepsilon:\Q_+}u\sim_\varepsilon
    v} \\
  &\qquad \ap_f(\eq(u,v,p))=\eq(f(u),f(v),\lambda\varepsilon.g(u,v,\varepsilon,p(\varepsilon))\big).\\
\end{align*}
\end{definition}

The role of the perhaps arbitrary-looking distance laws is that they couple the
behavior of $\rat$, $\lim$, and the pseudometric $\sim$.  They constitute a bare
minimum of data to formulate an induction principle.
\begin{remark}\label{rem:cauchy:structure}\leavevmode
  \begin{enumerate}
  \item Identity maps are Cauchy structure morphisms, and Cauchy structure
    morphisms are closed under composition.
  \item The distance laws of a Cauchy structure are automatically preserved by
    Cauchy structure morphisms, as $\sim$ is valued in propositions.
  \item A morphism of Cauchy structures from $R$ to $S$ lifts to
    a map $\CA_{R}\to\CA_{S}$ on the Cauchy approximations.
  \item We emphasize that even though a Cauchy structure has the
    $\lim$ map, it need not be Cauchy complete, since the elements
    $x_\varepsilon$ of a Cauchy approximation might not be of the form
    $\rat(q)$ or $\lim(z)$.  In other words, the $\lim$ map does not
    necessarily compute limits.

    For example, we may define a Cauchy structure on a type $\bool$
    with two elements, where both $\rat$ and $\lim$ constantly output
    $\false$, and we have the relations $\true\sim_\varepsilon\true$
    and $\false\sim_\varepsilon\false$ for all $\varepsilon$, but
    nothing else.  Then we have a Cauchy approximation that is
    constantly $\true$, and $\lim$ computes it limit as $\false$,
    which is not a limit in the sense of
    Definition~\ref{def:limits:premetric}---a valid limit would be
    $\true$.
  \end{enumerate}
\end{remark}

\subsection{HoTT book reals}
\label{sec:reals:hott}

We recall the definition of the HoTT book reals
$\RH$~\cite[Section~11.3]{hottbook}.  But for reasons of convenience, we will
use an algebraic definition centered around Cauchy structures.  This is
equivalent as we will show in Theorem~\ref{thm:hott-induction}, so that the
variant used here is justified.

Defining types inductively as a kind of homotopy-initial structure is a typical
approach in HoTT.  For instance, the circle may be defined as the
homotopy-initial structure among those types that have a point, and a path from
that point to itself, and where morphisms are defined as maps that preserve both
the point and the path.

\begin{definition}\label{def:hott:book:reals}
  $\RH$ is a homotopy-initial Cauchy structure, in the sense that for any other Cauchy
  structure $S$ (in any universe), the type of Cauchy structure morphisms from
  $\RH$ to $S$ is contractible.
\end{definition}

The equivalence of this definition with \hottbook\ comes from the fact that the
two definitions satisfy the same universal property.  More precisely, we will
develop an induction principle for $\RH$, so that $\RH$ is equivalent, and hence
by univalence, identical, to~\textcite[Section~11.3]{hottbook}.

However, we do not use the fact that Definition~\ref{def:hott:book:reals} is
equivalent to~\textcite[Section~11.3]{hottbook} in this paper.  The remainder of
this section is merely motivation for our definition.

\begin{definition}
  Given
  \begin{align*}
    A&:\RH\to\UU \\
    B&:\dpt{u,v:\RH}A(u)\to A(v) \to \dpt{\varepsilon:\Q_+}
    (u \sim_\varepsilon v) \to \Prop
  \end{align*}
  we obtain a natural premetric on
  $\sm{u:\RH}A(u)$, given by the relation:
  \[
    (u,a) \sim_\varepsilon (v,b) \coloneqq
      \sm{\zeta: u \sim_\varepsilon v} B(u,v,a,b,\varepsilon,\zeta)
  \]
\end{definition}

For the remainder of this section, fix a choice of $A:\RH\to\UU$ and
$B:\dpt{u,v:\RH}A(u)\to A(v) \to
\dpt{\varepsilon:\Q_+}(u \sim_\varepsilon v)\to\Prop$ --- these type
families will be input for our induction principle.  The remaining
input will allow us to define a Cauchy structure on $\sm{u:\RH} A(u)$.
We will often denote the type $B(u,v,a,b,\varepsilon,\zeta)$ by
$a\sim_\varepsilon b$, since $u$ can typically be inferred from $a$
and $v$ from $b$, and $\zeta$ is unique since the premetric on $\RH$
is valued in propositions.

\begin{definition}
  Let $x:\CA_{\RH}$ and
  $a:\dpt{\varepsilon:\Q_+}A(x_\varepsilon)$, satisfying
  \[
    \fa{\delta,\varepsilon:\Q_+}
    a_\delta\sim_{\delta+\varepsilon}a_\varepsilon.
  \]
  Then we call $a$ a \emph{dependent Cauchy approximation} over $x$.  We denote
  the type of all dependent Cauchy approximations over $x$ by
  $\mathcal{D}_{A}^x$, and again implicitly coerce its elements to their
  underlying (dependent) function.
\end{definition}
\begin{lemma}
  Suppose $x:\CA_{\RH}$ and
  $a:\dpt{\varepsilon:\Q_+}A(x_\varepsilon)$.
  Then the function
  \[
    \lambda \varepsilon . (x_\varepsilon , a_\varepsilon)
  \]
  is a Cauchy approximation in $\sm{u:\RH}A(u)$ iff $a$ is
  a dependent Cauchy approximation over $x$.
\end{lemma}
\begin{proof}
  Straightforward.
\end{proof}

The above lemma allows us to take limits componentwise, as we will do
in the proof of an induction principle in
Theorem~\ref{thm:hott-induction}.  To be able to phrase an induction
principle, we first define dependent identifications, namely the
identity of elements in a type family evaluated at identical elements
of $\RH$.

\begin{definition}\label{def:dependent:path}
  Given a type $A:\UU$, a type family $B:A\to\UU$, an identification
  $p:x=_Ay$ in $A$, and elements $u:B(x)$ and $v:B(y)$, the type of
  \emph{dependent identifications}
  $u=_p^Bv$ is defined by induction on $p$: if $p$ is $\refl(x)$ then
  $(u=_{\refl(x)}^B v) \coloneqq (u=_{B(x)}v)$. We refer to elements
  of $u=_p^Bv$ as \emph{identifications from $u$ to $v$ over
    $p$}.
\end{definition}

In particular, an identification $p:x=_Ay$ can be combined with a
dependent identification $q:u=_p^B v$ into an identification
$(x,u)=_{\sm{a:A}B(a)}(y,v)$ in the dependent sum type, and vice versa
an identification in the dependent sum type gives rise to an
identification $p$ in $A$ and a dependent identification over $p$.

\begin{theorem}\label{thm:hott-induction}
  Suppose we are provided
  \begin{align*}
    A&:\RH\to\UU \\
    B&:\dpt{u,v:\RH}A(u)\to A(v) \to \dpt{\varepsilon:\Q_+}
    (u \sim_\varepsilon v) \to \Prop
  \end{align*}
  and the following data.
  \begin{align*}
    f_{\rat}:\;&\dpt{q:\Q}A(\rat(q)) \\
    f_{\lim}:\;&\dpt{x:\CA_{\RH}}
             \mathcal{D}_{A}^x \to A(\lim(x))
    \\
    f_{\eq}:\;&\dpt{u,v:\RH}
              \dpt{a:A(u)}
              \dpt{b:A(v)}
              \dpt{p:\fa{\varepsilon:\Q_+}a\sim_\varepsilon b}
              a=_{\eq(u,v,p)}^A b
              \displaybreak[1]\\
    f_{d_{\rat,\rat}}:\;&
      \dpt{q,r:\Q}\dpt{\varepsilon:\Q_+}-\varepsilon<q-r<\varepsilon
      \to f_{\rat}(q)\sim_\varepsilon f_{\rat}(r)
      \displaybreak[2]\\
    f_{d_{\rat,\lim}}:\;  &
      \dpt{q:\Q} \dpt{y:\CA_{\RH}}
      \dpt{b:\mathcal{D}_{A}^y}
      \dpt{\delta,\varepsilon:\Q_+}
      \rat(q)\sim_{\varepsilon} y_\delta\\
      & \quad\quad\quad\quad
      \to f_{\rat}(q)\sim_{\varepsilon}b_\delta \to
      f_{\rat}(q) \sim_{\varepsilon+\delta} f_{\lim}(y,b)
      \displaybreak[3]\\
    f_{d_{\lim,\rat}}:\;  &
      \dpt{x:\CA_{\RH}}
      \dpt{a:\mathcal{D}_{A}^x}
      \dpt{r:\Q}
      \dpt{\delta,\varepsilon:\Q_+}
      x_\delta\sim_{\varepsilon} \rat(r)\\
      & \quad\quad\quad\quad
      \to a_\delta\sim_{\varepsilon} f_{\rat}(r) \to
      f_{\lim}(x,a) \sim_{\varepsilon+\delta} f_{\rat}(r)
      \displaybreak[3]\\
    f_{d_{\lim,\lim}}:\;  &
      \dpt{x,y:\CA_{\RH}}
      \dpt{a:\mathcal{D}_{A}^x}
      \dpt{b:\mathcal{D}_{A}^y}
      \dpt{\delta,\eta,\varepsilon:\Q_+}
      x_\delta\sim_{\varepsilon} y_\eta\\
      & \quad\quad\quad\quad
      \to a_\delta\sim_{\varepsilon}b_\eta \to
      f_{\lim}(x,a) \sim_{\varepsilon+\delta+\eta} f_{\lim}(y,b)
  \end{align*}
  In that case, we obtain
  \begin{align*}
    f:{}&\dpt{u:\RH} A(u) \qquad\text{and} \\
    g:{}&\dpt{u,v:\RH} \dpt{\varepsilon:\Q_+}
        \dpt{\zeta:x\sim_\varepsilon y} B(u,v,f(u),f(v),\varepsilon,\zeta),
  \end{align*}
  satisfying
  \begin{align*}
    f(\rat(q))&=f_{\rat}(q) \qquad\text{and} \\
    f(\lim(x))&=f_{\lim}(x,(f,g)[x]),
  \end{align*}
  where $(f,g)[x]$ is the dependent Cauchy approximation defined by
  \[
    (f,g){[x]}_\varepsilon\coloneqq f(x_\varepsilon).
  \]
\end{theorem}

\begin{proof}
  We reason similarly to \textcite{sojakova:hit}.  Write
  $T=\sm{u:\RH}A(u)$.  Given the input data, we can define a natural
  Cauchy structure on $T$.  For example,
  $\rat_T(q) \coloneqq (\rat(q),f_{\rat}(q))$.

  Hence, by homotopy-initiality of $\RH$, we obtain $h:\RH\to T$ and
  $i_{\varepsilon,u,v} : u \sim_\varepsilon v \to h(u)
  \sim_\varepsilon h(v)$ preserving $\rat$, $\lim$ and $\eq$ in the
  obvious sense.

  Postcomposing $h$ and $i$ (the latter componentwise) with the first projection
  functions $\proj1$ gives us a Cauchy morphism $\proj1\circ h : \RH\to\RH$.
  Now the uniqueness of the homotopy-initial map $\RH\to\RH$ gives us
  identifications of the form $\proj1(h(u))=u$ for $u:\RH$.

  Meanwhile, we have the second projection
  \[
    \proj2 \circ h : \dpt{u:\RH}A(\proj1(h(u)))
  \]
  which almost has the right type.  By transporting this dependent function
  pointwise along the identification $\proj1(h(u))=u$, and similarly for $i$, we
  obtain dependent functions $f$ and $g$ with the required type.

  From the fact that $h$ and $g$ form a Cauchy structure morphism, and thus
  preserve $\rat$ and $\lim$, we get that $f(\rat(q))=f_{\rat}(q)$, and similarly
  for $\lim$.
\end{proof}

We have shown that $\RH$ satisfies the same universal property as the
type defined in~\cite[Section~11.3]{hottbook}, so that the types are
equivalent.

Important properties of $\RH$ such as Cauchy completeness require some work.
For instance, the type of the $\lim$ map does not, on its own, guarantee that it
compute limits, as discussed in Remark~\ref{rem:cauchy:structure}.  However, we
may now import such important properties
from~\textcite[Section~11.3.2]{hottbook}, appealing to the fact that the types
are equivalent.

\subsection{Dedekind reals}
\label{sec:reals:dedekind}

A Dedekind real is defined by a pair $(L,U)$ of predicates $\pow\Q$ on $\Q$ with
some properties.  To phrase these properties succinctly, we use the following
notation for $x=(L,U)$:
\begin{align*}
  (q<x) &\coloneqq (q \in L) \qquad\text{and} \\
  (x<r) &\coloneqq (r \in U).
\end{align*}
This notation will be justified by the fact that $q\in L$ holds iff
$\rat(q)<x$, with $\rat$ the inclusion of the rationals into the Dedekind
reals, defined below.

\begin{definition}\label{def:dedekind}
  A pair $x=(L,U)$ of predicates on the rationals is a \emph{Dedekind
    cut} or \emph{Dedekind real} if it satisfies the four Dedekind
  properties:
  \begin{enumerate}
  \item \emph{bounded:} $\ex{q : \Q} q<x$ and
    $\ex{r : \Q} x<r$.
  \item \emph{rounded:} For all $q,r : \Q$,
    \begin{align*}
      q<x & \Iff \ex{q' : \Q} (q < q') \land (q'<x)
            \qquad\text{and}
      \\
      x<r & \Iff \ex{r' : \Q} (r' < r) \land (x<r').
    \end{align*}
  \item \emph{transitive:} $(q<x)\land (x<r)\Implies (q<r)$ for all
    $q, r : \Q$.
  \item \emph{located:} $(q < r) \Implies (q<x) \lor (x<r)$ for all
    $q, r : \Q$.
  \end{enumerate}

  The collection $\RD$ of pairs of predicates $(L,U)$
  together with proofs of the four properties, collected in a
  $\Sigma$-type, is called the \emph{Dedekind reals}.
\end{definition}
If we have $\Q:\UU$, and a hierarchy of universes $\UU:\VV$, then we may regard
$\RD$ as a type in $\VV$.

\begin{remark}
  \Textcite{hottbook} has
  \emph{disjointness}
  \[
    \fa{q:\Q}\neg(x<q\land q<x)
  \]
  instead of the transitivity property, which is equivalent to it in
  the presence of the other conditions, and it is this disjointness
  condition that we use most often in proofs.
\end{remark}
We will now endow $\RD$ with a Cauchy structure.
%
%
We define the embedding of the rationals, addition, subtraction, the inequality
relation, the absolute value function, and the premetric on $\RD$, for $q,r:\Q$
and $x,y:\RD$ and $\varepsilon:\Q_+$.  Note that following our notation for
$x=(L,U)$, the first 8 lines here are actually 4 pairs of definitions of a
Dedekind cut.  Well-definedness of these cuts needs to be checked, which we
refrain from doing here as they are (equivalent with) standard constructive
definitions.
\begin{eqnarray*}
  (q<\rat(r))&\coloneqq& q<r, \\
  (\rat(q)<r)&\coloneqq& q<r, \\
  (q<x+y)&\coloneqq&\exists(s,t:\Q).(q=s+t) \wedge (s<x) \wedge (t<y),\\
  (x+y<r)&\coloneqq&\exists(s,t:\Q).(r=s+t) \wedge (x<s) \wedge (y<t),\\
  (q<-x)&\coloneqq&x<-q,\\
  (-x<r)&\coloneqq&-r<x,\\
  (q<\abs{x})&\coloneqq&(q<x)\vee (q<-x),\\
  (\abs{x}<r)&\coloneqq&(x<r)\wedge (-x<r),\\
  (x<y)&\coloneqq&\exists(q:\Q).(x<q) \wedge (q<y),  \\
  x\sim_\varepsilon y&\coloneqq&\abs{x-y}<\varepsilon.
\end{eqnarray*}



% ...

\begin{theorem}\label{thm:dedekind-cauchy-complete}
  $\RD$ is Cauchy complete.
\end{theorem}
\begin{proof}

  Let $x:\mathcal{C}_\RD$.  We define $\lim(x):\RD$ by:
  \begin{eqnarray*}
    (q<\lim(x)) \coloneqq \exists (\varepsilon,\theta:\Q_+)
    . (q+\varepsilon+\theta < x_\varepsilon), \\
    (\lim(x)<r) \coloneqq \exists (\varepsilon,\theta:\Q_+)
    . (x_\varepsilon < r-\varepsilon-\theta ).
  \end{eqnarray*}
  The details of well-definedness of the map $\lim$, as well as a
  proof that it constructs limits, can be found in~\textcite[Theorem~11.2.12]{hottbook}.
\end{proof}

\begin{theorem}\label{thm:dedekind-cauchy-struct}
  $\RD$ and the previously constructed $\rat$ and $\lim$ can be
  completed to obtain a Cauchy structure.
\end{theorem}

For this we will use the following.

\begin{lemma}\label{lem:RD-lim-comm-addition}
  For $x:\mathcal{C}_\RD$ and $u:\RD$, we have $\lim(x+u)=\lim(x)+u$,
  where $x+u$ is the Cauchy approximation given by
  $\lambda\varepsilon.x_\varepsilon+u$.
\end{lemma}
\begin{proof}
  One can show, for example by first showing that the Dedekind reals
  form a group, that $x+u$ is a Cauchy approximation.  The remainder
  is straightforward by unrolling the definitions.
\end{proof}
\begin{proof}[Proof of Theorem~\ref{thm:dedekind-cauchy-struct}]
  The remaining ingredients are $\eq$ and verification of the four
  distance laws.

  First, $\eq$.  Suppose that $u\sim_\varepsilon v$ for arbitrary
  $\varepsilon:\Q_+$.  We need to show that $u=v$.  Without loss of generality,
  we show that for $q:\Q$, if $q<u$ then $q<v$.  By roundedness, there exists
  $q':\Q$ with $q<q'<u$.  By locatedness, $q<v \vee v<q'$.  In the left case
  $q<v$ we are done.  In the right case, again by roundedness, there exists
  $r:\Q$ with $v<r<q'$.  By the fact that $u\sim_{(q'-r)}v$, there exist
  $s,s':\Q$ with $s'<v<s+(q'-r)$ and $s<u<s'+(q'-r)$.  So in particular
  $s'<v<r$, and thus $s'<r$.  But at the same time $q'<u<s'+(q'-r)$, which
  yields $r<s'$.  So we have the contradiction that $s'<r$ and $r<s'$.

  The distance law $d_{\rat,\rat}$ follows from the definition of $\rat$ and the
  fact that $\rat(q-r)=\rat(q)-\rat(r)$.  The remaining three distance
  laws can be shown by applying Lemma~\ref{lem:RD-lim-comm-addition}.
\end{proof}

By checking the details of \textcite[Theorem~11.3.50]{hottbook}, we get the following:
\begin{proposition}\label{prop:reals:subsets}
  The natural embedding $i_H:\RH\to\RD$ is a Cauchy structure morphism.
\end{proposition}

\subsection{Euclidean reals}
\label{sec:euclidean-reals}

\Textcite{escardo:simpson:interval} showed that, in
any elementary topos, the Euclidean real interval is an interval
object.  They carried out the proof in a type theory for
toposes~\cite{lambek:scott,MacLaneS:shegl,johnstone:sketches},
higher-order intuitionistic logic, which we adapt to our type theory,
assuming propositional resizing.

\begin{definition}\label{def:euclidean:reals:pred}
  The type $\RE$ of \emph{Euclidean reals} is defined as least Cauchy complete
  subset of the Dedekind reals containing the rationals.  In other words, for
  every $R:\pow_\UU\RD$ (where $\UU$ is arbitrary) containing the rationals, and
  which is Cauchy complete, we have $\RE\subseteq R$.
\end{definition}

When we think of $\Prop_\UU$ as a collection of truth values, motivated
by the subobject classifier $\Omega$ in toposes, we may consider the
possibility that there is only one such collection.
\begin{definition}\label{def:prop-resiz}
  \emph{Propositional resizing} holds if for any two universes
  $\UU,\VV$, we have \[\Prop_\UU\eqv\Prop_\VV.\]
\end{definition}
In the presence of propositional resizing, we will consequently write $\Prop$
and $\pow A$, dropping the universe.

\begin{lemma}\label{lem:subtypes-complete}
  Assuming propositional resizing, for any $A:\UU$, the type
  $\pow A$ of subsets of $A$ is a complete lattice.  That is,
  for any collection of subsets $E:\pow\pow A$, the
  union $\bigcup E:\pow A$, defined using propositional
  resizing by
  \[
    \left(\bigcup E\right)(a)\coloneqq\ex{B:\pow A}B\in E\land
    a\in B,
  \]
  is a join of $E$, and similarly the intersection $\bigcap E:\pow A$ defined by
  \[
    \left(\bigcap E\right)(a)\coloneqq\fa{B:\pow A}B\in E\Implies
    a\in B
  \]
  is a meet of $E$.
\end{lemma}

\begin{lemma}[Escard\'o~and~Simpson~\cite{escardo:simpson:interval}]\label{lem:euclidean:reals:impred}
  Assuming propositional resizing, the type $\RE$ of
  \emph{Euclidean reals} can be constructed as the meet (as in
  Lemma~\ref{lem:subtypes-complete}) of the subtypes of the Dedekind
  reals which are Cauchy complete and contain the rationals.
\end{lemma}

With propositional resizing, the type of all sets in a given universe is a
topos~\cite[Theorem~10.1.12]{hottbook}, with $\Prop$ acting as a subobject
classifier.  This allows us to interpret Escard\'o and Simpson's definition, and
construction, of interval objects in toposes.
\begin{theorem}
  Assuming propositional resizing, so that we can construct $\RE$ as
  an element of some universe $\UU$.  The unit interval in $\RE$ is an
  interval object, where interval objects are defined as in Escard\'o
  and Simpson with respect to that category of sets in universe $\UU$.
\end{theorem}
The proof is simply a translation of the proof
in~\textcite{escardo:simpson:interval}, where we note that our definition of
$\RE$ coincides with the definition in category-theoretic terms.

\section{Assuming existence of $\RH$}
\label{sec:reals:euclidean}

In order to relate $\RH$ to $\RE$, without assuming propositional resizing, we
relate $\RH$ to an arbitrary Cauchy complete subset $R$ of $\RD$ that contains
the rationals, using the homotopy-initiality of $\RH$ as in
Definition~\ref{def:hott:book:reals}.  This yields a canonical embedding of
$\RH$ into $R$, in a more direct way than~\textcite[Theorem~11.3.50]{hottbook}.
So we reduce the question of coincidence of $\RH$ and $\RE$ to the fact that
both are minimal Cauchy complete subsets of the Dedekind reals, answering the
conjecture positively.

Let $R:\pow_\UU\RD$ be a subtype of the Dedekind reals.
Note that since the homotopy-initiality of $\RH$ is relative to a Cauchy
structure in \emph{any} universe, the choice of universe $\UU$ will not matter.
We can consider the collection $\sm{x:\RD}x\in R$ of elements in $R$, as in
Section~\ref{sec:subtypes}.  We restrict the Cauchy structure of $\RD$ obtained
from Theorem~\ref{thm:dedekind-cauchy-struct} to $R$.

\begin{proposition}\label{prop:subset:cauchy:struct}
  Given a Cauchy complete subset $R:\pow\RD$ of the Dedekind reals
  containing the rationals, the Cauchy structure on $\RD$ restricts to
  a Cauchy structure on $R$.
\end{proposition}
\begin{proof}
  First, the premetric on $R$ is inherited from the one on $\RD$ by
  restriction: for $\varepsilon:\Q_+$ and $x,y:\RD$ with $\mu:x\in R$
  and $\nu:x\in R$, we simply say that
  $(x,\mu)\sim_\varepsilon(y,\nu)$ holds iff $x\sim_\varepsilon y$.

  The map $\rat:\Q\to\RD$ constructed in
  Theorem~\ref{thm:dedekind-cauchy-struct} is an embedding, so that we may see
  $\Q$ as a subtype $\Q:\pow\RD$ of the Dedekind reals.  Assuming $R$ is a
  subtype of $\RD$ containing the rationals, i.e.\ $\Q\subseteq R\subseteq\RD$,
  we also get $\rat:\Q\to R$ by a straightforward restriction.

  In order to phrase when we have a $\lim$ structure, we define a
  subset $\CA_R$ of the type $\CA_\RD$, consisting of Cauchy
  approximations in $R$, by, for $x:\CA_\RD$,
  \[
    \CA_R(x)\coloneqq\fa{\varepsilon:\Q_+}x_\varepsilon\in R,
  \]
  noting that this $\CA_R$, now seen as a type that embeds into
  $\CA_\RD$, is equivalent to the type of Cauchy approximations in
  $\sm{x:\RD}{x\in R}$.  By further assuming that $R$ is Cauchy
  complete in the sense that for every Cauchy approximation
  $x\in\CA_R$ of elements in $R$, i.e.\ $x:\CA_\RD$ such that
  $\fa{\varepsilon:\Q_+}x_\varepsilon\in R$, there exists a limit of
  $x$ in $R$, we obtain a $\lim$ map: after all, we can compute the
  limit in $\RD$ using the $\lim$ structure of $\RD$, and then Cauchy
  completeness of $R$ states that this unique limit is an element of
  $R$.


  The construction of $\eq$ follows from Definition~\ref{def:embeddings}.  That
  is, the projection map $\proj1:(\sm{x:\RD}x\in R)\to\RD$ is an equivalence
  between identity types of $R$ and identity types of $\RD$, so that we may
  appeal to the $\eq$ structure of $\RD$.

  The distance laws hold because the premetric on $R$ is just the
  restriction of the premetric on $\RD$.
\end{proof}

\begin{corollary}\label{cor:euclidean-inclusion-cauchy}
  The map $i_R$ that includes $R$ into $\RD$ is a Cauchy structure morphism.
\end{corollary}

Proposition~\ref{prop:reals:subsets} established $\RH$ as a subset of
$\RD$ using a Cauchy structure morphism $i_H:\RH\to\RD$.  So we have
two subsets $\RH$ and $R$ of $\RD$.  The following proposition tells
us that $\RH\subseteq R$.

\begin{proposition}\label{prop:hott-euclidean-subset}
  We have $\RH\subseteq R$ as subsets of $\RD$.  That is, there is a
  horizontal map in the following diagram making the triangle commute.

  \begin{tikzcd}[column sep={1.5cm,between origins}]
    \RH
    \arrow[rr,dashrightarrow,yshift=0.6ex,"f"]
    \arrow[ddr,hookrightarrow,"i_H" left]
    &
    &
    R
    \arrow[ddl,hookrightarrow,"i_R" right]
    \\ & \node[yshift=1cm]{\bigcirc}; &
    \\
    &
    \RD
    &
  \end{tikzcd}
\end{proposition}
\begin{proof}
  By homotopy-initiality of $\RH$, we obtain
  $f:\RH\to R$, and by the fact that Cauchy structure morphisms are
  closed under composition, using homotopy-initiality of the Cauchy
  structure of $\RH$ once more, we obtain the commutativity condition
  $i_R\circ f = i_H$.
\end{proof}
Lemma~\ref{lem:triangle:embeddings} additionally yields that the map
$f:\RH\to R$ above is an embedding.

We have shown that $\RH\subseteq R$ for an arbitrary Cauchy complete subset $R$
of $\RD$ containing the rationals.  Thus, in conclusion:

\begin{corollary}
  $\RH$ satisfies Definition~\ref{def:euclidean:reals:pred} of $\RE$.
\end{corollary}

\section{Assuming propositional resizing}
\label{sec:assuming:prop}

In the previous section, we have related $\RH$ and $\RE$ by showing
that $\RH$ is the least Cauchy complete subset of $\RD$ containing the
rationals---a result that requires having the type $\RH$ in the first
place.  In a type theory where $\RH$ is not given as a primitive type,
we can still relate the Euclidean reals and the HoTT book reals.  The
HoTT book reals are defined uniquely by their universal property; that
is, any two homotopy-initial Cauchy structures are equal.  The goal of
this section is to show that $\RE$ satisfies that same universal
property, so that when we do have $\RH$, it coincides with $\RE$.

We borrow two strategies from the proof of
Escard\'o--Simpson~\cite{escardo:simpson:interval}
that the interval in the Euclidean reals is an interval object, namely
\begin{enumerate}
\item defining a dcpo, such that the construction of a certain point
  of that dcpo corresponds to proving the theorem, and
\item using a fixed point theorem, based on
  Pataraia's~\cite{pataraia:fixed:point}, to construct that point.
\end{enumerate}

\newcommand{\PR}{{\pow\RD}}
\newcommand{\F}{\mathcal{F}_{(S,\sim)}}

Concretely, we need to show that for any Cauchy structure $(S,\sim)$,
the type $\CShom(\RE,S)$ of Cauchy structure morphisms is
contractible.  So for a given Cauchy structure $(S,\sim)$, we define a
certain subdcpo $\F$ of $\PR$ whose elements are subsets
$\Q\subseteq R\subseteq\RE$ for which, loosely speaking, the type of
Cauchy structure morphisms \emph{restricted to $R$} is contractible.
By showing that $\RE$ is an element of $\F$, we have the required
result.  In particular, $\RE$ is found as a fixed point of a certain
$\F$-closed endomap $\Phi$, which extends a subset $R$ to the set of
limits of sequences valued in $R$.

The definitions of $\F$ and $\Phi$ loosely follow the style
of Escard\'o--Simpson, but have some changes since we are showing a
different universal property and working in a different logic.

The construction of $\F$ and $\Phi$, and establishing their required
properties, requires extensive calculations, since the construction of
an element of $\F$ requires showing that a certain type of restricted
Cauchy structure morphisms is contractible.  This contractibility, in
turn, consists of the construction of a restricted Cauchy structure
morphism, and a proof of uniqueness of those restricted Cauchy
structure morphisms.  The fact that the fixed point theorem that we
use has weaker assumptions than, for instance, Kleene's or
Knaster--Tarski's works to our advantage.

Although the proof of Pataraia's fixed point theorem would use the
propositional resizing axiom of Definition~\ref{def:prop-resiz}, we
use a weaker version, Corollary~\ref{cor:subdcpo:fixpoint}, which does not require it.
However, we do use propositional resizing to appeal to
Lemma~\ref{lem:subtypes:dcpo}, which gives that $\PR$ is a dcpo.  We
use the specific construction of the joins in our proof that $\F$ is a
subdcpo.

\subsection{Dcpos}
\label{sec:dcpos}

A general theory of dcpos is developed in~\textcite{jong:scott:model}, where
adequate universe levels are calculated in full detail.  We sidestep this by
simply assuming propositional resizing, in which case the topos-theoretic
approach works out as usual.

\begin{definition}\label{def:orders:pre:po}\leavevmode
  A \emph{partial order} is a set $X:\UU$ with a binary relation $R: X\to X\to \Prop_\UU$ which is:
  \begin{enumerate}
  \item \emph{reflexive}, i.e.
    $\fa{x:X}Rxx$;
  \item \emph{antisymmetric}, i.e.
    $\fa{x,y:X}Rxy \Implies Ryx\Implies
    x=y$;
  \item \emph{transitive}, i.e.
    $\fa{x,y,z:X}Rxy \Implies Ryz\Implies
    Rxz$;
  \end{enumerate}
\end{definition}

\newcommand{\D}{\mathcal{D}}

\newcommand{\subdcpo}{\operatorname{{subdcpo}}}
\begin{definition}
  Let $(A,\leq)$ be a partially ordered set.
  \begin{enumerate}
  \item An endomap $f:A\to A$ is \emph{inflationary} if it is
    \emph{monotonic}, i.e.\ $\fa{x,y:A}x\leq y\Implies f(x)\leq f(y)$,
    and \emph{increasing}, i.e. $\fa{x:A}x\leq f(x)$.
  \item A subset $\D:\pow A$ of $A$ is \emph{semidirected} if
    whenever $x,y\in \D$, there exists $z\in \D$ with $x\leq z$ and
    $y\leq z$.
  \item A subset $\D$ of $A$ is \emph{directed} if it is semidirected
    and inhabited.
  \item A partial order $(A,\leq)$ is a \emph{directed-complete partial order (dcpo)}
    if every directed subset $\D:\pow A$ of $A$ has a join in $A$, i.e.\ has
    an upper bound $w:A$ of $\D$ such that if $v$ is also an upper
    bound of $\D$, then $w\leq v$.
  \item A subset $B :\pow A$ of a dcpo $(A,\leq)$ is a \emph{subdcpo}
    if whenever $\D$ is a directed subset of $A$ contained in $B$, its
    join is contained in $B$.
  \item If we need to be precise about universe levels, for a given partial
    order $(A, \leq)$ with $A:\UU$, we should consider $\VV$-subsets
    $\D:\pow_\VV A$ to be \emph{$\VV$-semidirected} resp.\
    \emph{$\VV$-directed}, $(A, \leq)$ to be a \emph{$\VV$-dcpo} and
    $B:\pow_\VV A$ to be a \emph{$\VV$-subdcpo} of $A$.
  \end{enumerate}
\end{definition}
The following lemma justifies the name subdcpo.
\begin{lemma}\label{lem:subdcpo:to:dcpo}
  A $\VV$-subdcpo $B:\pow_\VV A$ of a $\VV$-dcpo $(A,\leq)$ gives rise
  to a $\VV$-dcpo $(\sm{b:A}b\in B,\leq)$ of elements in $B$ with the
  ordering given by restriction as
  \[
    (b,\mu)\leq(b',\mu')\coloneqq b\leq b'.
  \]
  Here we assume that $\UU\sqsubseteq\VV$.
\end{lemma}
\begin{proof}
  Let $\D:\pow_\VV(\sm{b:A}b\in B)$ be a directed subset of
  $\sm{b:A}b\in B$.  Now we see $\D$ as a directed subset of $A$, by defining $\D'$ as
  \[
    (d\in\D')\coloneqq \sm{\mu:d\in B}((d,\mu)\in \D)
  \]
  $\D'$ is directed because $\D$ is, and contained in the subdcpo $B$, so that
  it has a join in $B$, that is, $\bigvee \D'\in B$, which gives a join of $\D$
  in $\sm{b:A}b\in B$.
\end{proof}

Finally, since every complete lattice is a dcpo, we have our main example:
\begin{lemma}\label{lem:subtypes:dcpo}
  Assuming propositional resizing, for any $A:\UU$, the type
  $\pow A$ of subtypes of $A$ is a dcpo under the
  $\subseteq$ ordering.
\end{lemma}

\subsection{Fixed points}
\label{sec:dcpo:fixpoints}

Given a certain endomap $f:A\to A$ on a dcpo, we aim to construct a fixed point
of $f$.  Perhaps surprisingly, if we additionally have that $f$ is increasing,
so that it is inflationary, then we do not need propositional resizing to
compute a fixed point, and this is Corollary~\ref{cor:subdcpo:fixpoint} below.

\begin{proposition}[Pataraia~\cite{pataraia:fixed:point},
  Escard\'o--Simpson~\cite{escardo:simpson:interval}]\label{prop:pataraia:prop}
  Let $(A,\leq)$ be a $\UU$-dcpo with $A:\UU$.  The subset $I:\pow_\UU(A\to A)$ of
  $A\to A$ of inflationary endomaps, given by
  \[
    (f\in I)\coloneqq(\fa{x,y:A}x\leq y\Implies f(x)\leq
    f(y))\land(\fa{x:A}x\leq f(x)),
  \]
  is a $\UU$-subdcpo.  $I$ is a $\UU$-directed
  subset of $I$, so that $I$ has a top element $\top$.  Given a point
  $x:A$, $\top(x)$ is a common fixed point of all inflationary maps on
  $A$.
\end{proposition}
\begin{proof}
  Let $\D\subseteq I$ be directed.  To show that its join $\bigvee \D$
  in $A\to A$ is an inflationary map, notice that if $x\leq y$ in $A$
  then $\left(\bigvee \D\right)(y)$ is an upper bound of $\D[x]$, and
  that for $x:A$ and any $f\in \D$, we have $x\leq f(x)$, so that
  $\left(\bigvee \D\right)(x)$ is an upper bound of $f(x)$ and hence of
  $x$.

  $I$ is semidirected in $I$ because for $f,g\in I$ we have $f,g\leq f \circ g$
  where the latter is again inflationary.  It is inhabited because the identity
  map is inflationary.  Hence $I$ is directed.

  Let $x:A$ and let $f:A\to A$ be inflationary, so that in particular
  $\top \leq f \circ \top$.  Since $f\in I$, hence $f\circ\top\in I$, thus
  $f\circ \top \leq \top$, and hence $f \circ \top = \top$, making $\top(x)$ a
  fixed point of $f$.
\end{proof}
The following corollary is the fixed point theorem we will use in Section~\ref{sec:eucl-reals-satisfy}.
\begin{corollary}\label{cor:subdcpo:fixpoint}
  Let $(A,\leq)$ be a $\VV$-dcpo, and $f:A\to A$ an inflationary
  endomap.  If $B:\pow_\VV A$ is an $f$-closed subdcpo of $A$,
  then from a point of $B$ we can construct a fixed point of $f$.
\end{corollary}
\begin{proof}
  The type $\sm{b:A}b\in B$ of elements in $B$ is a
  $\VV$-dcpo by Lemma~\ref{lem:subdcpo:to:dcpo}, and $f:A\to A$ gives
  rise to an inflationary endomap on it, so that
  Proposition~\ref{prop:pataraia:prop} applies.
\end{proof}

\subsection{Quantification over subtypes}
\label{sec:quantifier:subtype}

Given a subtype $B:\pow A$ of $A$, we sometimes consider
\emph{only} the elements of $A$ that happen to be in $B$.  In other
words, we consider the elements of the type $\sm{b:A}b\in B$
corresponding via Lemma~\ref{lem:embeddings-subtypes} to $B$.  We
introduce the following notation.

\begin{definition}\label{def:notation:quant:subtype}
  For $A:\UU$, $B:\pow A$ and $C:\left(\sm{a:A}a\in B\right)\to\UU$
  and $D:(\sm{a:A}a\in B)\to\Prop$, we write
  \begin{align*}
    \dpt{b\in B}C(b)\coloneqq{}
    &\dpt{b:A}\dpt{\nu:b\in B}C(b,\nu),\\
    \sm{b\in B}C(b)\coloneqq{}
    &\sm{b:A}\sm{\nu:b\in B}C(b,\nu),\\
    \fa{b\in B}D(b)\coloneqq{}
    &\fa{b:A}\fa{\nu:b\in B}D(b,\nu),\\
    \ex{b\in B}D(b)\coloneqq{}
    &\ex{b:A}\ex{\nu:b\in B}D(b,\nu).\\
  \end{align*}
  For $C:A\to\UU$ and $D:A\to\Prop$, this simplifies to the notation
  \begin{align*}
    \dpt{b\in B}C(b)\coloneqq{}
    &\dpt{b:A}b\in B \to C(b),\\
    \sm{b\in B}C(b)\coloneqq{}
    &\sm{b:A}b\in B\times C(b),\\
    \fa{b\in B}D(b)\coloneqq{}
    &\fa{b:A}b\in B\Implies D(b),\\
    \ex{b\in B}D(b)\coloneqq{}
    &\ex{b:A}b\in B\land D(b).\\
  \end{align*}
  For $C:\UU$, this further simplifies to the notation for function
  types
  \begin{align*}
    B\to C\coloneqq{}
    &\dpt{b:A}b\in B \to C.
  \end{align*}
\end{definition}
\begin{remark}
  A different way to read the above notations is using the
  correspondence of Lemma~\ref{lem:embeddings-subtypes}, so that, for
  instance,
  \[
    \dpt{b\in B}C(b)\coloneqq\dpt{t:\sm{b:A}b\in B}C(t).
  \]
  It is straightforward to check that this type coincides with the
  above interpretation.
\end{remark}
\begin{remark}
  Following the correspondence of Lemma~\ref{lem:embeddings-subtypes},
  for $B:\pow A$ and $C:\UU$, we read $C\to B$ as the type
  $C\to\sm{b:A}b\in B$.
\end{remark}

\subsection{Homotopy-initiality of the Euclidean reals}
\label{sec:eucl-reals-satisfy}

\begin{theorem}\label{thm:euclidean:hinit}
  Assuming propositional resizing, the Euclidean reals satisfy the
  universal property of the HoTT book reals of
  Section~\ref{sec:reals:hott} for sets.  That is, for a Cauchy
  structure $(S,\sim)$, where $S$ is a set in any universe, the type $\CShom(\RE,S)$
  of Cauchy structure morphisms from $\RE$ to $S$ is
  contractible.
\end{theorem}
\begin{remark}
  Since $S$ is a set, saying that $\CShom(\RE,S)$ is contractible is
  equivalent to saying that there exists a Cauchy structure
  morphism from $\RE$ to $S$, and any two such morphisms are pointwise
  equal.
\end{remark}
It would be desirable to be able to prove homotopy-initiality for
arbitrary types $S$, rather than only for sets, but we leave this as
an open problem.  A similar issue arises in work by Awodey, Frey and
Speight on impredicative encodings of higher inductive
types~\cite{Awodey:impredicative}.

We refer to the data of the Cauchy structure on $\RD$ as $\rat$, $\lim$ and
$\eq$, and to the data of another Cauchy structure $(S,\sim)$ with subscripts as
$\rat_S$, $\lim_S$, $\eq_S$, $d_{\rat,\rat,S}$, $d_{\rat,\lim,S}$,
$d_{\lim,\rat,S}$ and $d_{\lim,\lim,S}$.

\newcommand{\dcpolevel}{{i+2}}
\begin{proof}
  As in Proposition~\ref{prop:subset:cauchy:struct}, for a given subset
  $Y:\PR$ of the Dedekind reals, we define a subset
  $\CA_Y:\pow\CA_\RD$ of the type $\CA_\RD$ of Cauchy
  approximations in $\RD$, with $x:\CA_\RD$, as
  \[
    \CA_Y(x)\coloneqq\fa{\varepsilon:\Q_+}x_\varepsilon\in Y.
  \]

  For a subset $Y$ of $\RD$ with $\Q\subseteq Y$, we define what it
  means to have a restricted Cauchy structure morphism $Y\to S$.
  Compared to ordinary Cauchy structure morphism as in
  Definition~\ref{def:cauchy:structure}, the essence of the definition
  is that although the output of $\rat:\Q\to\RD$ is always an element
  of $Y$, because $\Q\subseteq Y$, the output of $\lim:\CA_Y\to\RD$
  may not be, and so we require the corresponding preservation
  condition for $Y$ only in the case that it is.  Additionally,
  because $S$ is a set, preservation of the $\eq$ structure is
  automatic.  In conclusion, we define
  \begin{align*}
    \subCShom(Y,S)\coloneqq
    & \sm{f:Y\to S}\\
    & \sm{g:\dpt{\varepsilon:\Q_+}\dpt{u,v\in Y}u\sim_\varepsilon
      v\to f(u) \sim_\varepsilon f(v)}\\
    & (\dpt{q:\Q}f(\rat(q))=\rat_S(q)) \\
    \times
    & (\dpt{x\in\CA_Y}\lim x \in Y \Implies f(\lim x) =
      \lim_S(f\circ x))
  \end{align*}
  where, following Definition~\ref{def:notation:quant:subtype},
  $\dpt{u,v\in Y}C(u,v)$ means $\dpt{u,v:\RD}u,v\in Y\Implies C(u,v)$,
  and similarly $\dpt{x\in\CA_Y}D(x)$ means
  $\dpt{x:\CA_\RD}x\in\CA_Y\Implies D(x)$.

  Note that $\subCShom(\RE,S)\eqv\CShom(\RE,S)$ because $\lim$ is
  always defined on $\RE$.
  The goal is to show that $\RE$ is an element of the subset
  $\F:\pow\PR$ of $\PR$ defined by
  \begin{align*}
    \F(Y)\coloneqq\Q\subseteq Y\subseteq \RE \land
    \isContr(\subCShom(Y,S)),
  \end{align*}
  so that there is a unique Cauchy structure morphism from $\RE$ to
  $S$.  We show this by using Corollary~\ref{cor:subdcpo:fixpoint} to
  construct a fixed point of a certain map $\Phi$ that we will define
  later, and then showing that this fixed point is a Cauchy complete
  subset of $\RE$, so that it must be equal to $\RE$.

  Note that two restricted Cauchy structure morphisms, that is, two
  elements of the type $\subCShom(Y,S)$, are equal iff their
  underlying maps $Y\to S$ are equal, because the remaining data is a
  proposition.

  First, to be more precise, in order to be able to use
  Corollary~\ref{cor:subdcpo:fixpoint}, we show the following claims.
  \begin{claim}\label{cla:dcpo}
    $\PR$ is a dcpo with the relation $\subseteq$.
  \end{claim}
  \begin{claim}\label{cla:subdcpo}
    $\F$ is a subdcpo of $\PR$.
  \end{claim}
  \begin{claim}\label{cla:rat}
    $\Q\in\F$.
  \end{claim}
  \begin{claim}\label{cla:inflationary}
    The map $\Phi:\PR\to\PR$, which we define later, is inflationary.
  \end{claim}
  \begin{claim}\label{cla:closed}
    $\F$ is $\Phi$-closed.
  \end{claim}

  \begin{proof}[Proof of Claim~\ref{cla:dcpo}]
    By Lemma~\ref{lem:subtypes:dcpo}, indeed $\PR$ is a dcpo.
  \end{proof}
  \begin{proof}[Proof of Claim~\ref{cla:subdcpo}]
    To show that $\F$ is a subdcpo, let $\D:\pow\PR$ with $\D\subseteq \F$ be a
    directed subset of $\F$.  Following Lemma~\ref{lem:subtypes-complete}, the
    join of $\D$ in $\PR$ is constructed using propositional resizing as
    $Y\coloneqq\bigcup \D$ with $Y:\PR$, and we claim that it is an element of
    $\F$.  Because $\D\subseteq\F$, the various elements $X\in \D$ come equipped
    to their own restricted Cauchy structure morphism which is unique on $X$,
    and we refer to their underlying maps as $f_X:X\to S$ and
    $g_X:\dpt{\varepsilon:\Q_+}\dpt{u,v\in X}u\sim_\varepsilon v\to f_X (u)
    \sim_\varepsilon f_X (v)$.

    $\Q\subseteq Y \subseteq \RE$ holds because the elements of $\D$
    satisfy this property, and $\D$ is inhabited.

    To show that $\subCShom(Y,S)$ is a proposition, consider two
    restricted Cauchy structure morphisms with maps $f,f':Y\to S$, and
    let $y\in Y$, recalling from
    Definition~\ref{def:notation:quant:subtype} that this means we
    take $y:\RD$ and assume $y\in Y$.  We aim to show the proposition
    $f(y)=f'(y)$, so we may assume to have $X\in \D$ with $y\in X$.
    Both $f$ and $f'$ restrict to restricted Cauchy structure morphisms
    on $X$, where they must both equal the center of contraction given
    by $f_X:X\to S$, and in particular $f(y)=f_X(y)=f'(y)$.

    It remains to find an element of $\subCShom(Y,S)$.

    We construct $f_Y : Y \to S$ as a certain map
    \[
      f_Y':\dpt{y\in Y}\sm{s:S}\ex{X\in \D} y\in X \land f_X (y) = s
    \]
    composed with a projection map that forgets the proof of
    $\ex{X\in \D} y\in X \land f_X (y) = s$.  Notice that for every
    $y$, the codomain $\sm{s:S}\ex{X\in \D} y\in X \land f_X (y) = s$
    of $f_Y'$ is a proposition, because given $s,s':S$ and
    $X,X'\in \D$ with $y\in X$ and $y\in X'$ and $f_X(y)=s$ and
    $f_{X'}(y)=s'$, from the fact that $\D$ is directed, we know that
    there exists $Z\in \D$ with $X,X'\subseteq Z$.  But the map $f_Z$
    restricts to both $X$ and $X'$ where it must be equal to $f_X$ and
    $f_{X'}$, respectively, so that $s=f_X(y)=f_Z(y)=f_{X'}(y)=s'$.

    To construct $f_Y'$, take an element $y\in Y$.  By the construction of
    $Y=\bigcup\D$, this means there exists $X\in \D$ with $y\in X$.  Since the
    codomain is a proposition, we may assume to \emph{have} $X\in \D$ with
    $y\in X$.  Hence we can take $f'_Y(y)$ to be given by $s\coloneqq f_X(y)$.

    To construct
    $g_Y:\dpt{\varepsilon:\Q_+}\dpt{u,v\in Y}u\sim_\varepsilon v\to
    f_Y(u)\sim_\varepsilon f_Y(v)$, let $\varepsilon:\Q_+$, let
    $X,X'\in \D$ with $u\in X$ and $v\in X'$, and let
    $\nu:u\sim_\varepsilon v$.  Because $\D$ is directed, we know that
    there exists $Z\in \D$ with $X,X'\subseteq Z$, so that we can
    output $g_Z(\varepsilon,u,v,\nu)$.

    To show the preservation conditions, first note that
    $\fa{X\in \D}\fa{x \in X}f_Y(x)=f_X(x)$, because $f_Y(x):S$ and
    $f_X(x):S$ both arise as elements of the codomain
    $\sm{s:S}\ex{X\in \D} {y\in X} \land f_X (y) = s$ of $f_Y'$, which
    is a proposition, as shown above.

    To show that $\dpt{q:\Q}f_Y(\rat(q))=\rat_S(q)$, let $q:\Q$.
    Since $\D$ is inhabited, there exists $X\in \D$, and since we are
    showing a proposition, we may assume to have such an $X$.  Then,
    because $f_X$ satisfies the preservation conditions,
    $f_Y(\rat(q))=f_X(\rat(q))=\rat_S(q)$.

    To show that
    \( \dpt{x\in\CA_Y}\lim x \in Y \Implies f_Y(\lim x) =
    \lim_S(f_Y\circ x) \), let $x\in\CA_Y$ and assume $\lim x \in Y$.

    One may be inclined to look for $X\in \D$ with
    $x_\varepsilon\in X$ for all $\varepsilon:\Q_+$, and also
    $\lim x \in X$, suggesting that we need $\D$ to be
    infinitary-directed, meaning that we would have an element in $\D$
    which contains all $x_\varepsilon$.  In fact, we can avoid this by
    observing that $\lim x$ can be computed as the limit of the
    constant Cauchy approximation $\lambda\varepsilon'.\lim x$.  If we
    can show that the Cauchy approximation $f_Y\circ x$ is close to
    the constant Cauchy approximation
    $\lambda\varepsilon'.f_Y(\lim x)$, then we can use $\eq_S$ to show
    the required preservation condition.  We now make this argument
    more precise.

    First, note that since $\RD$ is Cauchy complete indeed we have
    $\lim x=\lim(\lambda\varepsilon'.\lim x)$.  Since $\lim x\in Y$,
    and since we are showing a proposition, we may assume to have
    $X\in \D$ with $\lim x \in X$.  Then
    \begin{align*}
      f_Y(\lim x)
      &=f_X(\lim x)\\
      &=f_X(\lim(\lambda\varepsilon'.\lim x))\\
      &=\lim_S(\lambda\varepsilon'.f_X(\lim x))\\
      &=\lim_S(\lambda\varepsilon'.f_Y(\lim x)).
    \end{align*}
    By $\eq_S$, it suffices to show
    \[
      \fa{\varepsilon:\Q_+}\lim_S(\lambda\varepsilon'.f_Y(\lim
      x))\sim_\varepsilon\lim_S(f_Y \circ x).
    \]
    Let $\varepsilon:\Q_+$.  The distance law $d_{\lim,\lim,S}$ gives us,
    with $\varepsilon/2$, $\varepsilon/4$ and $\varepsilon/4$ respectively for
    $\varepsilon$, $\delta$ and $\eta$:
    \[
      f_Y(\lim x) \sim_{\varepsilon/2} f_Y(x_{\varepsilon/4})\to
      \lim_S(\lambda\varepsilon'.f_Y(\lim x)) \sim_{\varepsilon}
      \lim_S(f_Y\circ x).
    \]
    In order to show the proposition
    \( f_Y(\lim x)\sim_{\varepsilon/2}f_Y(x_{\varepsilon/4}) \), from
    directedness of $\D$ we obtain $X'\in \D$ with $\lim x \in X'$ and
    $x_{\varepsilon/4}\in X'$.  Then
    \[
      g_{X'}(\varepsilon/2,\lim x, x_{\varepsilon/4}):\lim
      x\sim_{\varepsilon/2}x_{\varepsilon/4} \to f_{X'}(\lim
      x)\sim_{\varepsilon/2}f_{X'}(x_{\varepsilon/4} ),
    \]
    and $\lim x\sim_{\varepsilon/2}x_{\varepsilon/4}$ can be shown
    using Cauchy completeness (as in
    Definition~\ref{def:limits:premetric}) of $\RD$.

    This concludes the proof of Claim~\ref{cla:subdcpo} that $\F$ is
    a subdcpo of $\PR$.
  \end{proof}
  \begin{proof}[Proof of Claim~\ref{cla:rat}]
    To show that $\Q\in\F$, note that $\Q\subseteq\Q\subseteq\RE$.  To
    show that the type $\subCShom(\Q,S)$ is a proposition, let
    $f,f':\Q\to S$ be Cauchy structure morphisms.  Since they both
    satisfy the preservation condition for rationals, we have
    $f(\rat(q))=\rat_S(q)=f'(\rat(q))$, as required.

    Now we construct an element of $\subCShom(\Q,S)$.  The map $f_\Q:\Q\to S$ is
    given by $\rat_S$ directly.  Then $g_\Q$ can be constructed using the
    distance law $d_{\rat,\rat,S}$.  The preservation condition for rationals
    holds by definition.  Let $x\in\CA_\Q$ and assume $\lim x\in \Q$.  We need
    to show $f(\lim x)=\lim_S(f\circ x)$, i.e.\
    $\rat_S(\lim x)=\lim_S(f\circ x)$.  So by $\eq_S$ and the second distance
    law $d_{\rat,\lim,S}$ it suffices to show for arbitrary $\varepsilon:\Q_+$
    that $\rat_S(\lim x)\sim_{2\varepsilon/3}f(x_{\varepsilon/3})$, i.e. that
    $\rat_S(\lim x)\sim_{2\varepsilon/3}\rat(x_{\varepsilon/3})$, i.e. by
    $d_{\rat,\rat,S}$ that
    $-2\varepsilon/3<\lim x-x_{\varepsilon/3}<2\varepsilon/3$, which holds
    because $\lim x$ is a limit of $x$.
  \end{proof}
  \begin{proof}[Proof of Claim~\ref{cla:inflationary}]
    We now define an inflationary $\F$-closed map $\Phi$ whose fixed
    point we will show to be $\RE$.

    For $X:\PR$, define $\Phi(X):\PR$ to be the subset of $\RD$ of
    limits of Cauchy approximations valued in $X$, that is:
    \[
      \Phi(X)(y)\coloneqq\ex{x\in\CA_X}y=\lim x
    \]
    The map $\Phi$ is increasing because every real is the limit of a
    constant sequence, and monotone because if $X\subseteq Y$ then
    $\CA_X\subseteq\CA_Y$.
  \end{proof}
  \begin{proof}[Proof of Claim~\ref{cla:closed}]
    To show that $\Phi$ is $\F$-closed, assume $X\in\F$, and note that
    $\Q\subseteq\Phi(X)$ holds because
    $\rat(q)=\lim(\lambda\varepsilon'.\rat(q))$, and
    $\Phi(X)\subseteq\RE$ follows from Cauchy completeness of $\RE$.

    From $\isContr(\subCShom(X,S))$ we obtain $f_X:X\to S$.

    To show that $\subCShom(\Phi(X),S)$ is a proposition, consider two
    restricted Cauchy structure morphisms with maps
    $f,f':\Phi(X)\to S$, and let $x\in\CA_X$.  To show the proposition
    $f(\lim x)=f'(\lim x)$, note that $f$ and $f'$ restrict to the
    same restricted Cauchy structure morphism $f_X$ on $X$.  Then,
    since $f$ and $f'$ satisfy the preservation condition for the
    limit of $x$, we have
    $f(\lim x)=\lim_S(f\circ x)=\lim_S(f'\circ x)=f'(\lim x)$.

    We define $f_{\Phi(X)}:\Phi(X)\to S$ as a certain map
    \[
      f_{\Phi(X)}':\dpt{y\in\Phi(X)}\sm{s:S}\ex{x\in\CA_X}y=\lim x
      \land s=\lim_S (f_X\circ x)
    \]
    followed by a projection map that forgets the proof of
    $\ex{x\in\CA_X}s=\lim_S (f_X\circ x)$.

    First we show that the codomain of $f_{\Phi(X)}'$ is a
    proposition.  For suppose $s,s':S$ and $x,x'\in\CA_X$ with
    $y=\lim x=\lim x'$ and $s=\lim_S(f_X \circ x)$ and
    $s'=\lim_S(f_X\circ x')$.  Because $\lim$ computes limits in
    $\RD$, we know that $\lim x=\lim x'$ implies
    \[
      \fa{\varepsilon,\varepsilon',\theta,\theta':\Q_+}x_\varepsilon\sim_{\varepsilon+\varepsilon'+\theta+\theta'}x_{\varepsilon'}
    \]
    and hence in particular
    \[
      \fa{\varepsilon:\Q_+}x_{\varepsilon/6}\sim_{4\varepsilon/6}x'_{\varepsilon/6}.
    \]
    Now $g_X$ gives us
    \[
      \fa{\varepsilon:\Q_+}f_X(x_{\varepsilon/6})\sim_{4\varepsilon/6}f_X(x'_{\varepsilon/6})
    \]
    and so by the distance law $d_{\lim,\lim,S}$
    \[
      \fa{\varepsilon:\Q_+}\lim_S(f_X\circ x)\sim_{\varepsilon}\lim_S(
      f_X\circ x')
    \]
    and so with the $\eq_S$, we get
    $s=\lim_S(f_X\circ x)=\lim_S(f_X\circ x')=s'$, as required.

    Since the codomain of $f_{\Phi(X)}'$ is a proposition, for a given
    $y\in\Phi(X)$ we may assume to have $x\in\CA_X$ with $y=\lim x$.
    Then we can compute the output as $\lim_S(f_X\circ x)$, completing
    the definition of $f_{\Phi(X)}'$ and $f_{\Phi(X)}:\Phi(X)\to S$.

    To define
    \[
      g_{\Phi(X)}:\dpt{\varepsilon:\Q_+}\dpt{u,v\in
        \Phi(X)}u\sim_\varepsilon v\to f_{\Phi(X)} (u)
      \sim_\varepsilon f_{\Phi(X)} (v),
    \]
    let $\varepsilon:\Q_+$, $x,y\in\CA_X$ and
    $\nu:\lim x\sim_\varepsilon\lim y$.  We aim to show
    $f_{\Phi(X)}(\lim x)\sim_\varepsilon f_{\Phi(X)}(\lim y)$, i.e.\
    $\lim_S(f_X\circ x)\sim_\varepsilon \lim_S(f_X \circ y)$ by the
    above definition of $f_{\Phi(X)}$.

    Since $\lim x\sim_\varepsilon\lim y $, that is,
    $\abs{\lim x-\lim y}<\varepsilon$, by the Archimedean property we
    know that
    \[
      \ex{\delta:\Q_+}\abs{\lim x-\lim y}<\delta<\varepsilon.
    \]
    Since we are showing a proposition, we may assume to have such a
    $\delta$.  Because $\lim$ computes limits in $\RD$, using the
    definition of $\sim$ in $\RD$ we know that
    $\lim x\sim_\delta\lim y$ implies
    \[
      \fa{\xi,\xi',\theta,\theta':\Q_+}x_\xi\sim_{\delta+\xi+\xi'+\theta+\theta'}y_{\xi'}
    \]
    and so in particular with
    $\xi\coloneqq\frac{\varepsilon-\delta}{6}$ we have
    \( x_\xi\sim_{\delta+4\xi}y_{\xi} \).  Then $g_X$ gives
    \(f_X(x_\xi)\sim_{\delta+4\xi}f_X(y_{\xi}) \) and hence by the
    fourth distance law
    \( \lim_S(f_X\circ x)\sim_\varepsilon \lim_S(f_X \circ y) \).

    To show that $f_{\Phi(X)}$ and $g_{\Phi(X)}$ satisfy the coherence
    conditions for restricted Cauchy structure morphisms, let $q:\Q$.
    Then $f_{\Phi(X)}(\rat(q))$ may be computed as
    $\lim_S(\lambda\varepsilon'.f_X(\rat(q)))$, which, by the fact
    that $f_X$ is a restricted Cauchy structure morphism, is equal to
    $\lim_S(\lambda\varepsilon'.\rat_S(q))$.  By $\eq_S$ it suffices
    to show
    \[
      \fa{\varepsilon:\Q_+}\lim_S(\lambda\varepsilon'.\rat_S(q))\sim_\varepsilon\rat_S(q),
    \]
    so let $\varepsilon:\Q_+$.  Then
    $\rat_S(q)\sim_{\varepsilon/2}\rat_S(q)$ by the first distance
    law, and so
    $\lim_S(\lambda\varepsilon'.\rat_S(q))\sim_{\varepsilon}\rat_S(q)$
    by the distance law $d_{\lim,\rat,S}$.

    For the second preservation condition, let $x\in\CA_{\Phi(X)}$ and
    assume $\lim x\in\Phi(X)$, that is,
    $\ex{x'\in\CA_X}\lim x =\lim x'$.  Since we are showing the
    proposition $f_{\Phi(X)}(\lim x)=\lim_S(f_{\Phi(X)}\circ x)$, let
    $x'$ be such, so that we have to show
    $\lim_S(f_X\circ x')=\lim_S(f_{\Phi(X)}\circ x)$.  By $\eq_S$, it
    suffices to show for $\varepsilon:\Q_+$ that
    \[
      \lim_S(f_X\circ x')\sim_\varepsilon\lim_S(f_{\Phi(X)}\circ x).
    \]
    Using the distance law $d_{\lim,\lim,S}$, it suffices to show
    \[
      f_X(x'_{\varepsilon/6})\sim_{4\varepsilon/6}f_{\Phi(X)}(x_{\varepsilon/6}).
    \]
    Now
    \[
      f_X(x'_{\varepsilon/6})=f_X(\lim(\lambda\varepsilon'.x'_{\varepsilon/6}))=\lim_S(f_X\circ(\lambda\varepsilon'.x'_{\varepsilon/6}))=f_{\Phi(X)}(\lambda\varepsilon'.x'_{\varepsilon/6})=f_{\Phi(X)}(x'_{\varepsilon/6}),
    \]
    so this is equivalent to
    \[
      f_{\Phi(X)}(x'_{\varepsilon/6})\sim_{4\varepsilon/6}f_{\Phi(X)}(x_{\varepsilon/6})
    \]
    and so by $g_{\Phi(X)}$ it suffices to show
    \[
      x'_{\varepsilon/6}\sim_{4\varepsilon/6}x_{\varepsilon/6}
    \]
    which holds because $\lim x =\lim x'$.

    This concludes the proof of Claim~\ref{cla:closed} that $\Phi$ is $\F$-closed.
  \end{proof}
  Hence, by Corollary~\ref{cor:subdcpo:fixpoint}, $\Phi$ has a fixed
  point $R$ in $\F$.  By definition, $R\subseteq\RE$.  It
  remains to show that $\RE\subseteq R$, which will follow from the
  fact that $R$ is a Cauchy complete subset of the Dedekind reals
  containing the rationals.  Additionally, the fact that $R$ contains
  the rationals is part of the definition of $\F$, so we
  only have to show that $R$ is Cauchy complete.

  Let $x\in\CA_R$.  By definition, we have $\lim x \in\Phi(R)$.
  Since $R$ is a fixed point of $\Phi$, we have $\lim x\in R$, and
  this is a limit of $x$ because it is its limit in $\RD$.

  Hence $\RE$ is an element of $\F$.
\end{proof}

Finally, we consider what happens in the absence of propositional resizing.

The construction of $\RE$ will not go through, as our construction of an
intersection of subsets $\pow_\UU\RD$ of $\RD$ results in a subset
$\pow_{\VV}\RD$, with $\UU:\VV$.  This is not a true meet because, living in the
wrong universe, it has the wrong type.

But suppose given any Cauchy complete $\R:\pow_\UU\RD$ containing the rationals,
which is the least such subset, can we prove a homotopy-initiality theorem
similar to Theorem~\ref{thm:euclidean:hinit}, replacing instances of $\RE$ with
$\R$?  The fixed point theorem, Corollary~\ref{cor:subdcpo:fixpoint}, that we
used in the proof of Theorem~\ref{thm:euclidean:hinit}, does not use
propositional resizing, and we also do not need it to construct the desired
fixed point $\R$ since we simply assume it to be given.  We cannot
straightforwardly apply Corollary~\ref{cor:subdcpo:fixpoint}, since we cannot
show $\PR$ to be a dcpo.  It may suffice to see $\PR$ as a partial order, and
$\F$ as a subdcpo of that partial order in the sense that it contains all the
joins of directed subsets \emph{that exist in $\PR$}.  Showing that $\F$ is a
subdcpo in this sense would still require the construction of a restricted
Cauchy structure morphism with some underlying maps $f_Y$ and $g_Y$ for a join
$Y$ of a directed subset $\D$ as in the proof above.  In the absence of
propositional resizing, we can not construct $Y$ using existential quantifiers
as in Lemma~\ref{lem:subtypes-complete}, and so the construction of $f_Y$ and $g_Y$
in the proof of Theorem~\ref{thm:euclidean:hinit} will not go through.

\section{Conclusion}
\label{sec:universal:open-questions}

Thanks to our systematic use of Cauchy structures, we have written a rather
short proof that the HoTT book reals coincide with the Euclidean
reals in Proposition~\ref{prop:hott-euclidean-subset}, without relying
on propositional resizing.

In the presence of propositional resizing, we can define $\RE$.
Theorem~\ref{thm:euclidean:hinit}, showing that $\RE$ is a
homotopy-initial Cauchy structure, without assuming that $\RH$ exists,
is new.  Two open questions remain in regard to this result:
\begin{enumerate}
\item Can we show homotopy-initiality with respect to arbitrary types
  equipped with Cauchy structures, rather than only sets?  Note that a
  type with a Cauchy structure is not automatically a set: given a
  Cauchy structure on a type $X$, we can assign a Cauchy structure to
  $X+Y$ for an arbitrary type $Y$, with elements in the right disjunct
  being assigned an infinite distance to all elements.
\item What homotopy-initiality can be shown in the absence of
  propositional resizing, given only a least Cauchy complete subset of
  the Dedekind reals, without knowing its construction as an
  intersection of subsets of $\RD$?
\end{enumerate}














{
  \printbibliography[heading=bibintoc]
}

\end{document}